\documentclass{article}
\usepackage[T1]{fontenc}
\usepackage[utf8]{inputenc}
\usepackage{float}
\usepackage{mathtools}
\usepackage{amsmath}
\usepackage{amsthm}
\usepackage{amssymb}
\usepackage{geometry}
\usepackage[authoryear,round]{natbib}

\makeatletter

\floatstyle{ruled}
\newfloat{algorithm}{tbp}{loa}
\providecommand{\algorithmname}{Algorithm}
\floatname{algorithm}{\protect\algorithmname}

\@ifundefined{date}{}{\date{}}
\usepackage{algorithmic} 
\usepackage{microtype}
\usepackage{graphicx}
\usepackage{subfig}
\usepackage{booktabs} 
\usepackage{xcolor}
\usepackage{hyperref}

\usepackage{amsmath}
\usepackage{amssymb}
\usepackage{mathtools}
\usepackage{amsthm}
\theoremstyle{plain}

\theoremstyle{definition}

\theoremstyle{remark}

\usepackage[textsize=tiny]{todonotes}
\renewcommand{\epsilon}{\varepsilon}
\renewcommand{\ln}{\log}
\renewcommand{\tilde}{\widetilde}

\makeatother

\theoremstyle{plain}
\newtheorem{thm}{\protect\theoremname}
\theoremstyle{remark}
\newtheorem{rem}[thm]{\protect\remarkname}
\theoremstyle{plain}
\newtheorem{lem}[thm]{\protect\lemmaname}
\newtheorem{prop}[thm]{\protect\propositionname}
\theoremstyle{remark}
\newtheorem{claim}[thm]{\protect\claimname}
\providecommand{\claimname}{Claim}
\providecommand{\lemmaname}{Lemma}
\providecommand{\propositionname}{Proposition}
\providecommand{\remarkname}{Remark}
\providecommand{\theoremname}{Theorem}

\begin{document}
\global\long\def\R{\mathbb{R}}%
\global\long\def\diag{\mathrm{diag}}%
\global\long\def\opt{\mathsf{OPT}}%
\global\long\def\smax{\text{smax}}%
\global\long\def\smin{\text{smin}}%
\global\long\def\po{\textsc{PackingOracle}}%
\global\long\def\co{\textsc{CoveringOracle}}%
\global\long\def\init{\text{init}}%
\global\long\def\pc{\textsc{MPCOracle}}%
\global\long\def\est{\textsc{MaxEstimator}}%

\title{Solving Positive Linear Programs with Differential Privacy}
\author{Alina Ene\thanks{Department of Computer Science, Boston University, \texttt{aene@bu.edu}.}\and 
Huy L. Nguyen\thanks{Khoury College of Computer Sciences, Northeastern University, \texttt{hu.nguyen@northeastern.edu}.}\and 
Ta Duy Nguyen\thanks{Department of Computer Science, Boston University, \texttt{taduy@bu.edu}.}\and 
Adrian Vladu\thanks{CNRS \& IRIF, Universit\'e Paris Cit\'e, \texttt{vladu@irif.fr}.}}
\maketitle
\begin{abstract}
We study differentially private approximation algorithms for positive
linear programs (LPs with nonnegative coefficients and variables),
focusing on the fundamental families of packing, covering, and mixed
packing-covering formulations. We focus on the high-sensitivity, constraint-private
regime of Hsu-Roth-Roughgarden-Ullman (ICALP 2014), where neighboring
instances may differ by an arbitrary single constraint, so one cannot
hope to approximately satisfy every constraint under privacy. We give
private solvers that return approximate solutions while violating
only a controlled number of constraints. Our algorithms improve the
prior instance-dependent guarantees, and also yield new data-independent
bounds that depend only on the dimension. Our techniques involve a
dense multiplicative weights update method developed from a regularized
dual viewpoint, which we analyze in a way that exploits structure
specific to positive LPs.

\end{abstract}

\section{Introduction}

Positive linear programs (LPs)---most notably packing, covering and
mixed packing-covering LPs---play a central role in machine learning
and computer science, where they are used to model optimization problems
with all non-negative variables and constraints. Packing LPs which
arise in problems such as resource usage maximization subject to capacity
constraints, have a wide range of applications including ad allocation,
bandwidth allocation and fractional matching, while covering LPs seen
in problems such as cost minimization subject to covering requirements
naturally model facility location relaxations, set cover, etc. Mixed
packing-covering LPs combining both types of constraints appear in
network flow control, learning with budget and fair allocation problems.
A key reason for the importance of this class of LPs is that they
admit fast approximation algorithms, including multiplicative weights
update, primal-dual methods, which scale to massive datasets common
in modern ML systems. 

In many scenarios, linear programs directly involve users' sensitive
data such as health or financial information. For example, in a packing
LP for bandwidth allocation, each constraint may encode a user's capacity
or usage limit while in a covering LP modeling minimum-cost provisioning
of services, each user's constraint may represent a personal service
requirement. In such cases, solving the LP requires the algorithm
designer to pay particular attention to protecting the users' private
information. This is where differential privacy proves its usefulness.
Differential privacy provides a principled framework for protecting
users in this context by ensuring that the output of an LP-solving
algorithm does not change significantly when the input LP differs
by one constraint.

Research on solving linear programs under differential privacy dates
back to \citet{hsu2014privately} and has since been extended to broader
settings and related problems, including private learning of subspaces
and halfspaces \citet{DBLP:conf/focs/BunNSV15,DBLP:conf/colt/BeimelMNS19,kaplan2020private,DBLP:conf/citc/GaoS21,DBLP:conf/nips/Ben-EliezerMZ22},
which are fundamental learning theory problems. A key challenge in
privately solving linear programs is an inherent impossibility result:
it is not possible to simultaneously guarantee differential privacy
and satisfy all constraints. Indeed, the addition or removal of a
single constraint can change an LP from feasible to infeasible or
vice versa, revealing sensitive information. Consequently, any differentially
private algorithm must allow for the violation or removal of some
constraints. Characterizing upper bounds on the number of constraints
that must be dropped to preserve privacy has therefore become a central
question in this line of work, including this paper.

To address this challenge, \citet{hsu2014privately} introduces a
general framework for approximately solving linear programs under
differential privacy. However, in this work, explicit constraint-private
algorithms and performance guarantees are provided only for the fractional
set cover problem. Moreover, the resulting bounds are instance-dependent:
the number of constraints that must be dropped depends on properties
of the problem such as the optimal objective value. More recent works
by \citet{kaplan2024differentially,ene2025solving} develop algorithmic
approaches for solving general LPs exactly. These methods achieve
data-independent guarantees, but incur high-degree polynomial dependence
on the problem dimension, which can be overly pessimistic for structured
instances such as positive LPs.

In this work, we revisit the problem of solving linear programs with
differential privacy, focusing on the class of positive LPs, including
packing, covering and mixed packing-covering LPs. We propose new approximation
algorithms for solving these problems as well as new analysis that
can exploit their structure to improve the number of constraints that
are violated due to privacy. 

\subsection{Our contribution}

We focus on constraint-private LPs where neighboring inputs can differ
by a single constraint. This is the most challenging setting among
the various scenarios introduced by \citet{hsu2014privately}, for
which guarantees for private approximation algorithms are not well
understood. 

We give constraint-private algorithms that output solutions satisfying
all constraints of the LP approximately, except for a bounded number
of constraints. We provide upper bounds on the number of constraints
that might be violated, distinguishing between two types of guarantee:
one that only depends on the problem dimension and one that involves
data-dependent properties, specific to the LP given as input. 

Our results are stated in the following theorem (precise statements
are given in subsequent sections). 
\begin{thm}
(cf. Theorems \ref{thm:packing} and \ref{thm:covering}) Given the
approximation factor $\alpha\in(0,1)$, there exist $(\epsilon,\delta)$-differentially
private algorithms that output a solution to a packing (covering)
LP given by constraints $Ax\le1$ (respectively, $Ax\ge1$) for $A\in\R^{n\times d}_{\ge0}$.
With high probability, the output solution satisfies $Ax\le1+\alpha$
(respectively, $Ax\ge1-\alpha$), except for at most $s$ constraints
where,
\[
s=\tilde{O}\left(\frac{\left(A_{\max}\opt\right)^{1.5}}{\alpha^{2}\epsilon}\right),\text{ or }s=\tilde{O}\left(\frac{d^{1.5}}{\alpha^{3.5}\epsilon}\right),
\]
where $\opt$ is the optimal objective of the LP and $A_{\max}$ denotes
the maximum entry of $A$.
\end{thm}

\begin{thm}
(cf. Theorem \ref{thm:mixed}) Given the approximation factor $\alpha\in(0,1)$,
there exist $(\epsilon,\delta)$-differentially private algorithms
that output a solution to a mixed packing-covering LP given by constraints
$Px\le1,Cx\ge1$ where $P\in\R^{p\times d}_{\ge0},C\in\R^{c\times d}_{\ge0}$.
With high probability, the output solution satisfies $Px\le1+\alpha,Cx\ge1-\alpha$,
except for at most $s$ constraints where,
\begin{align*}
s & =\tilde{O}\left(\frac{P_{\max}C_{\max}\sqrt{P_{\max}+C_{\max}}V^{2.5}}{\alpha^{4.5}\epsilon}\right),\\
\text{ or }s & =\tilde{O}\left(\frac{d^{3}}{\alpha^{6}\epsilon}\right)
\end{align*}
and $V=1^{\top}x$ for some feasible solution $x$, $P_{\max},C_{\max}$
denote the maximum entries of $P$ and $C$, respectively.
\end{thm}

In these theorems statements, the $\tilde{O}$-notation hides logarithmic
factors that depend on $n,d,\delta$, the success probability and
the range of the input. 
\begin{rem}
Our $d$-dependent bounds from these theorems are most useful in the
regime where the number of constraints $n$ is much bigger than the
problem dimension $d$. A natural example is the problem of setting
the prices (the variables) of some common goods among a large user
base to maximize revenue while making sure that most users can meet
their needs. Here, the number of users $n$ is vast, while the number
of variables $d$ remains small. Another example is the combinatorial
public project problem where the goal is to select public projects
from a set of $d$ candidates, so as to maximize the total utility
of $n$ users where $n\gg d$. This is a discrete problem but an LP
relaxation can still yield valuable information such as by rounding.
\end{rem}

\paragraph{Comparison with \citet{hsu2014privately}.}

While\textbf{ }\citet{hsu2014privately} provide a general framework
to solve constraint-private LPs, explicit bounds are only provided
for fractional set cover problem for which an explicit private oracle
can be constructed. For this problem where $A_{\max}=1$, the number
of violated constraints given by \citet{hsu2014privately}'s algorithm
is $\tilde{O}\left(\frac{\opt^{2}}{\alpha^{2}\epsilon}\right)$. By
comparison, for this problem, our algorithm gives a bound of $\tilde{O}\left(\frac{\opt^{1.5}}{\alpha^{2}\epsilon}\right)$.
This is an improvement by a factor $\opt^{0.5}$. At the same time,
our work goes beyond this instance-specific bound and provides an
upper bound of $\tilde{O}\left(\frac{d^{1.5}}{\alpha^{3.5}\epsilon}\right)$.
This bound improves the former when $\opt\gg\frac{d}{\alpha}$. 

\paragraph{Comparison with \citet{kaplan2024differentially,ene2025solving}.}

These works focus on solving general LPs with high precision with
differential privacy, i.e, achieving $\log$-dependence instead of
polynomial dependence on $\frac{1}{\alpha}$. For general LPs, these
solvers have high degree polynomial dependence on the dimension. Specifically,
\citet{kaplan2024differentially} have an upper bound of $\tilde{O}\left(\frac{d^{9}}{\epsilon}\right)$
and the best known bound from\textbf{ }\citet{ene2025solving} is
$\tilde{O}\left(\frac{d^{4}}{\epsilon}\right)$. In the special case
when the LP has a strictly positive margin $\rho$$,$ these bounds
can be improved to $\tilde{O}\left(\frac{d^{5}}{\epsilon}\mathrm{poly}\log\frac{1}{\rho}\right)$
and $\tilde{O}\left(\frac{d^{2}}{\epsilon}\mathrm{poly}\log\frac{1}{\rho}\right)$,
respectively. When we can afford to tolerate a low-precision approximation,
which is equivalent to creating an LP with margin $\alpha$, the algorithm
by\textbf{ }\citet{ene2025solving} achieves an upper bound of $\tilde{O}\left(\frac{d^{2}}{\epsilon}\mathrm{poly}\log\frac{1}{\alpha}\right)$
on the number of violated constraints. By comparison, our data-independent bounds are $\tilde{O}\left(\frac{d^{1.5}}{\alpha^{3.5}\epsilon}\right)$
for pure problems and $\tilde{O}\left(\frac{d^{3}}{\alpha^{6}\epsilon}\right)$
for mixed packing-covering LPs. For pure LPs, our algorithms improve
the dependence on the problem dimension, while having a worse bound
for the mixed packing-covering LPs. On the other hand, our algorithms
also offer instance-specific bound, which can improve these data-independent
bounds in certain regimes.

\paragraph{Our technique.}

Our main technique is a newly developed Dense Multiplicative Weights
Update algorithm for positive LPs. The prior work by \citet{hsu2014privately}
also uses a version of Dense Multiplicative Weights Update, relying
on the Bregman projection approach by \citet{DBLP:journals/jmlr/HerbsterW01}
as a blackbox. However, this blackbox approach fails to leverage the
positivity of the constraints. Instead, we develop our new toolset
from first principles, which allow us to exploit the structure of
the problem. 

Our algorithm assigns and update weights to the constraints in the
LP and makes sure that these weights do not exceed a certain threshold
in order to preserve the privacy when a constraint of the LP changes
(thresholding guarantees a bound on the sensitivity). We develop this
algorithm from a regularized dual viewpoint, using a variant of the
standard softmax function, where the dual variable has its coordinates
bounded by a specified threshold. More precisely, the main tool is
the truncated softmax function, defined as follows: 
\[
\text{smax}^{U}\left(x\right)=\max_{r\in\mathcal{D}^{U}}\left\langle x,r\right\rangle -\omega\left(r\right),
\]
where the domain of the dual variable ${\cal D}^{U}=\left\{ r\in\Delta_{n}:\max r\leq U\right\} $
is the unit simplex with coordinates truncated at some value $0<U<1$
and $\omega$ is the negative entropy. The weights of the constraints
are given by $\nabla\text{smax}^{U}\left(x\right)$. The main idea
when using $\text{smax}^{U}\left(x\right)$ is that bounding the regularization
term can provide an improved bound on the sensitivity of the LP. Properties
of this function and its gradient allow us control the contribution
of any set of $s=\frac{1}{U}$ constraints, which exactly provides
the guarantee on the number of violated constraints in all cases.

\subsection{Related Work}

\citet{hsu2014privately} were the first to initiate the study of
solving linear programs with differential privacy. In their work,
multiple definitions of neighboring inputs are considered. We only
consider the most challenging case of constraint-private LPs where
neighboring inputs can differ by a single constraint, which is the
same setting studied by \citet{kaplan2024differentially,ene2025solving}.
The main tool developed by \citet{hsu2014privately} is private MWU
to approximately solve LPs. Previously, private MWUs have been developed
in the context of linear queries release \cite{hardt2010multiplicative}
to achieve optimal accuracy and runtime. For solving LPs, the aim,
on the other hand, is to find a solution that can minimize the number
of violated constraints. \citet{kaplan2024differentially,ene2025solving}
use a different approach based on privatizing rescaled perceptron
algorithms to solve general LPs to high accuracy.

Our work focus on the class of positive LPs. Where privacy is not
a concern, there has been a long line of research developing multiplicative
weights update for positive linear programs to improve runtime over
algorithms for general LPs, starting from \citet{shahrokhi1990maximum,DBLP:journals/mor/PlotkinST95,grigoriadis1994fast,luby1993parallel,young2001sequential}
followed by many subsequent works (see the survey by \citet{arora2012multiplicative}).
It is therefore natural to seek analogous improvement when privacy
is imposed. Our work show that improvement is indeed possible.

\section{Preliminaries}\label{sec:Preliminaries}

\subsection{Notation}

For a matrix $A$, we use $A_{\min}$ and $A_{\max}$ to denote the
minimum and maximum entries of $A$. We use $A_{i}$ to denote the
$i$-th row. We use a scalar $a\in\R$ to denote the vector whose
coordinates are all $a$. The dimension can be inferred from the context.

\subsection{Positive Linear Programs}

We consider the problem of solving linear programs with positive entries.
There are three main families of such LPs.

\paragraph{Packing LPs.}

We seek to approximately solve the problem 
\[
\max_{x\ge0}c^{\top}x\text{ s.t }Ax\le b,
\]
where $A\in\R^{n\times d}_{\ge0},b\in\R^{n}_{\ge0},c\in\R^{d}_{\ge0}$.
Without loss of generality, we can assume that $b=1$ and $c=1$. 

\paragraph{Covering LPs.}

We seek to approximately solve the problem 
\[
\min_{x\ge0}c^{\top}x\text{ s.t }Ax\ge b,
\]
where $A\in\R^{n\times d}_{\ge0},b\in\R^{n}_{\ge0},c\in\R^{d}_{\ge0}$.
We will also assume that $b=1$ and $c=1$. 

\paragraph{Mixed Packing-Covering LPs.}

This is the most general positive LPs, for which we want to find $x\ge0$
such that
\[
Px\le b\text{ and }Cx\ge c,
\]
where $P\in\R^{p\times d}_{\ge0},C\in\R^{c\times d}_{\ge0},b\in\R^{p}_{\ge0},c\in\R^{c}_{\ge0}$.
We will assume that $b=1$ and $c=1$. 

\subsection{Differential Privacy}

We will represent an LP by its constraint matrix, after normalization
so that the objective coefficients and scalars are $1$s. We say that
two LPs with constraint matrices $A$ and $A'$ are neighbors if they
differ by a single constraint. A randomized algorithm ${\cal A}$
is said to be $(\epsilon,\delta)$-differentially private (DP) if
for all neighboring LPs $A$ and $A'$ and every subset of possible
outcomes ${\cal O}$,
\begin{align*}
\Pr\left[{\cal A}(A)\in{\cal O}\right] & \le e^{\epsilon}\Pr\left[{\cal A}(A')\in{\cal O}\right]+\delta.
\end{align*}
In the case $\delta=0$, we say the algorithm is $\epsilon$-DP. 

We will use the exponential mechanism \cite{DBLP:conf/focs/McSherryT07}---a
standard tool in differential privacy. The exponential mechanism involves
a score function $Q(i,A)$ which takes as input a candidate $i$ from
a finite set $S$ and a dataset $A$ and outputs a real value. Given
a dataset $A$, the exponential mechanism $\mathsf{EM}_{\epsilon,Q}$
outputs $i\in S$ with probability proportional to $\exp\left(\epsilon Q(i,A)\right)$.

We have the following guarantee.
\begin{thm}
The exponential mechanism is $\epsilon$-differentially private and
guarantee that for a score function $Q$ with sensitivity $\Delta$,
with probability at least $1-\beta$, 
\begin{align*}
Q(\mathsf{EM}_{\epsilon,Q}(A),A)\ge\max_{j\in S}Q(j,S)-\frac{2\Delta\left(\log d+\log\frac{1}{\beta}\right)}{\epsilon} & ,
\end{align*}
where we say a score function $Q$ has sensitivity $\Delta$ if for
all $i\in S$ and any neighboring inputs $A$ and $A'$, $\left|Q(i,A)-Q(i,A')\right|\le\Delta$. 
\end{thm}

\subsection{Tool: Truncated Softmax and Softmin }

The main tool we use to analyze our algorithm is the truncated softmax
function, defined as follows. Let $\Delta_{n}=\left\{ r\in\R^{n}_{\ge0}:\sum_{i}r_{i}=1\right\} $
be the unit simplex. Given a scalar parameter $0<U\leq1$, define
$\mathcal{D}^{U}=\left\{ r\in\Delta_{n}:\max r\leq U\right\} $. We
define 
\[
\text{smax}^{U}\left(x\right)=\max_{r\in\mathcal{D}^{U}}\left\langle x,r\right\rangle -\omega\left(r\right),
\]
where $\omega\left(r\right)=\sum_{i}r_{i}\ln r_{i}$ is the negative
entropy. We also define the truncated softmin function $\smin^{U}$
as
\begin{align*}
\smin^{U}\left(x\right) & =-\smax^{U}\left(-x\right)=\min_{r\in\mathcal{D}^{U}}\left\langle x,r\right\rangle +\omega\left(r\right).
\end{align*}

\paragraph{Gradient.}

The gradient of $\text{smax}^{U}\left(x\right)$ is given by
\begin{align*}
\left[\nabla\smax^{U}\left(x\right)\right]_{i} & =\min\left\{ U,\exp\left(x_{i}-t_{U}(x)\right)\right\} ,
\end{align*}
where $t_{U}(x)$ is such that $\nabla\text{smax}^{U}\left(x\right)\in\Delta_{n}$
i.e, $\sum_{i}\left[\nabla\text{\ensuremath{\smax}}^{U}\left(x\right)\right]_{i}=1$.
Similarly, 
\begin{align*}
\left[\nabla\smin^{U}\left(x\right)\right]_{i} & =\min\left\{ U,\exp\left(-x_{i}-t_{U}(x)\right)\right\} ,
\end{align*}
where $t_{U}(x)$ is such that $\nabla\text{\ensuremath{\smin}}^{U}\left(x\right)\in\Delta_{n}$
i.e, $\sum_{i}\left[\nabla\smin^{U}\left(x\right)\right]_{i}=1$.

\paragraph{Bounding the $\protect\smax^{U}$ and $\text{\ensuremath{\protect\smin}}^{U}$
increase.}

We bound the change in $\smax^{U}$ when performing an update $x\gets x+u$
for a vector $u\in\R^{n}_{\ge0}$
\[
\text{smax}^{U}\left(x+u\right)\leq\text{smax}^{U}\left(x\right)+\frac{e^{D}-1}{D}\left\langle \nabla\text{smax}^{U}\left(x\right),u\right\rangle ,
\]
where $D=\max_{i}\left\{ u_{i}\right\} $. For $u\in\R^{n}_{\ge0}$
such that $\max_{i}\left\{ u_{i}\right\} \le1$, we also have
\begin{equation}
\text{smax}^{U}\left(x+u\right)\leq\text{smax}^{U}\left(x\right)+\left(1+D\right)\left\langle \nabla\text{smax}^{U}\left(x\right),u\right\rangle .\label{eq:smax-update}
\end{equation}
Similarly,
\begin{align}
\text{\ensuremath{\smin}}^{U}\left(x+u\right) & \ge\text{\ensuremath{\smin}}^{U}\left(x\right)+\left(1-D\right)\left\langle \nabla\text{\ensuremath{\smin}}^{U}\left(x\right),u\right\rangle .\label{eq:smin-update}
\end{align}

\section{Private Algorithms for Packing LP}

In this section, we seek to approximately solve the following problem
with differential privacy 
\[
\max_{x\ge0}\ 1^{\top}x\text{ s.t }Ax\le1,
\]
where $A\in\R^{n\times d}_{\ge0}.$ We assume to know the value $\opt=\max\left\{ 1^{\top}x,Ax\le1\right\} $,
so the goal is, given the approximation factor $\alpha$, to find
$x$ such that 
\[
1^{\top}x\ge\left(1-\alpha\right)\opt\text{ and }Ax\le1+\alpha.
\]
We also assume the entries of the input matrix $A$ come from a bounded
range and are upper bounded by a known value $S=O\left(A_{\max}\right)$.
\begin{rem}
\label{rem:assumption}The assumption that we know the value of $\opt$
is also made in\textbf{ }\citet{hsu2014privately} and implicitly
in \citet{kaplan2024differentially,ene2025solving}. Here, the implicit
assumption is that the maximum coefficient for each variable $x_{j}$
across all constraints lies within a range $[m,M]$ for known values
$M>m>0$. One then has $\opt\in[\frac{1}{dM},\frac{d}{m}]$. We can
run the algorithm for each multiple of $\left(1+\alpha\right)$ in
this range which covers a $\left(1\pm\alpha\right)$ approximation
of $\opt$. Then using the exponential mechanism, we can select the
best guess of $\opt$ with only $\log$ factors incurred. By the same
reason, using a binary search in $[m,M]$, we can estimate the value
$S=O\left(A_{\max}\right)$. We will ignore these $\log$ factors
due to the binary searches in our bounds.
\end{rem}

\subsection{Algorithm}

\begin{algorithm}
\caption{Private Algorithm for Packing LP}

\label{alg:packing}

\begin{algorithmic}[1]

\STATE \textbf{Input}: $A\in\R^{n\times d}$, upper bound $S=O\left(A_{\max}\right)$
for the entries of $A$, optimal objective $\opt=\max\left\{ 1^{\top}x,Ax\le1\right\} $,
$\alpha,\beta\in(0,1)$, privacy parameters $\epsilon,\delta$.

\STATE Let $H=\frac{2d}{\alpha\cdot\opt}$ if pre-processing else
$H=S$

\STATE Let $T=\frac{20H\cdot\opt\ln n}{\alpha^{2}},$ $\epsilon'=\frac{\epsilon}{2\sqrt{T\log\frac{1}{\delta}}}$,
$s=\frac{60H\cdot\opt(\log d+\log\frac{T}{\beta})}{\alpha\epsilon'}$,
$U=\frac{1}{s}$

\STATE \textbf{(Optional) Pre-processing}: for entry $A_{ij}$ in
$A$: $A_{ij}=\min\left\{ A_{ij},H\right\} $

\STATE Initialize $x_{0}=0$

\STATE for $t=0,\dots,T-1$:

\STATE $\qquad$Let $\Delta_{t}=\po(x_{t},\epsilon')$

\STATE $\qquad$$x_{t+1}=x_{t}+\Delta_{t}$

\STATE Let $\overline{x}=\frac{x_{T}}{T}$

\STATE \textbf{(Optional) Post-processing}: for $j\in[d]$: if $\overline{x}_{j}\le\frac{2}{H},$
let $\overline{x}_{j}=0$

\STATE \textbf{Output} $\overline{x}$

\STATE $\po(x,\epsilon)$:

\STATE $\qquad$Let $\eta=\frac{\alpha}{10H\cdot\opt}$

\STATE $\qquad$Use the Exponential mechanism with privacy parameter
$\epsilon$ to find coordinate $j\in[d]$ with score function $Q(j)=-\left\langle \nabla\smax^{U}\left(\eta Ax\right),A1_{j}\right\rangle \opt$

\STATE $\qquad$\textbf{Output} $\Delta=1_{j}\cdot\opt$

\end{algorithmic}
\end{algorithm}

Our algorithm is presented in Figure \ref{alg:packing}. Starting
with the initial solution $x=0$, he algorithm proceeds iteratively
as follows. In each iteration, the algorithm calls $\po$, which uses
the exponential mechanism to find a coordinate $j$ such that $\left\langle \nabla\smax^{U}\left(\eta Ax\right),A1_{j}\right\rangle $
is minimized to form the update vector $\Delta_{t}.$ The output of
the algorithm is the average of the cumulative updates, i.e, $\overline{x}=\frac{\sum_{t}\Delta_{t}}{T}$. 

Algorithm \ref{alg:packing} is closely related to the Dense Multiplicative
Weights Update algorithm by \citet{DBLP:journals/jmlr/HerbsterW01}
which is employed by \citet{hsu2014privately} as a blackbox in their
private algorithm for solving LPs. One key distinction here is that
our Dense MWU algorithm exploits the structure of the positive LPs
in maintaining that in each update, the increase in the constraint
values $\text{smax}^{U}\left(\eta A\left(x_{t+1}\right)\right)-\text{smax}^{U}\left(\eta Ax_{t}\right)$
is small when updating $x_{t+1}=x_{t}+\Delta_{t}$ with an appropriate
step size $\eta$. More specifically, without the constraint positivity,
\citet{hsu2014privately} bound
\begin{align*}
\text{smax}^{U}\left(\eta A\left(x_{t+1}\right)\right)-\text{smax}^{U}\left(\eta Ax_{t}\right)\lesssim & \left\langle \nabla\text{smax}^{U}\left(\eta Ax_{t}\right),\eta A\Delta_{t}\right\rangle +\eta^{2}\left\Vert A\Delta_{t}\right\Vert ^{2}_{\infty}.
\end{align*}
The term $\left\Vert A\Delta_{t}\right\Vert _{\infty}$ results in
the convergence rate $\frac{\rho^{2}\log n}{\alpha^{2}}$ of Dense
MWU for general LP, where $\rho$ is the width of the problem, being
an upper bound for $\left\Vert A\Delta_{t}\right\Vert _{\infty}$.
On the other hand, if we have the constraint positivity, we can bound
\begin{align*}
\text{smax}^{U}\left(\eta A\left(x_{t+1}\right)\right)-\text{smax}^{U}\left(\eta Ax_{t}\right)\lesssim & \left(1+\eta\left\Vert A\Delta_{t}\right\Vert _{\infty}\right)\left\langle \nabla\text{smax}^{U}\left(\eta Ax_{t}\right),\eta A\Delta_{t}\right\rangle ,
\end{align*}
and improve the convergence rate to $\frac{\rho\log n}{\alpha^{2}}$.
The improvement of the runtime of the non-private algorithm by a factor
$\rho$ leads to the improvement in the number of violated constraints
in our private algorithm. Intuitively, every iteration requires some
privacy loss so having fewer iterations allows us to utilize the privacy
budget more efficiently.

At the same time, in each iteration, $\po$ guarantees the objective
increases quickly ($1^{\top}\Delta_{t}=\opt$) and the average of
the constraints $Ax\le1$ weighted by $\nabla\smax^{U}\left(\eta Ax_{t-1}\right)$,
i.e $\left\langle \nabla\smax^{U}\left(\eta Ax_{t-1}\right),A\Delta_{t}\right\rangle $
is satisfied approximately. Here, $\nabla\smax^{U}\left(\eta Ax_{t-1}\right)$
can be thought of as the projection of weights $\nabla\smax\left(\eta Ax_{t-1}\right)$
onto a dense distribution over the simplex $\Delta_{n}$ where no
weights exceed $U$ in order to achieve privacy. The use of the truncated
softmax function $\text{smax}^{U}$ allows us to control the sensitivity
of the score function $Q$ used in the exponential mechanism, which
directly translates to the number of constraints that have to be dropped.
This function also gives the bound
\begin{align*}
\max_{S\in\left[n\right]:\left|S\right|=s=\frac{1}{U}}\left\langle \frac{1_{S}}{\left|S\right|},A\overline{x}\right\rangle  & \le\text{smax}^{U}\left(A\overline{x}\right)
\end{align*}
which implies if $\overline{x}$ satisfies $\text{smax}^{U}\left(A\overline{x}\right)\le1+\alpha$,
the average of the top $\frac{1}{U}$ constraints cannot exceed $(1+\alpha)$
and thus the number of violated constraints is at most $\frac{1}{U}$.

It is optional to execute pre-processing and post-processing steps
to obtain the novel data-independent guarantee for the algorithm.
In the pre-processing step, the algorithm truncates the entries of
the input matrix $A$ by threshold $H=\frac{2d}{\alpha\cdot\opt}$.
This step makes sure that we can have an instance-independent bound
for the sensitivity of the score function used in the exponential
mechanism. To make sure that the output solution satisfies the original
constraints instead of the new constraints resulted from the pre-processing,
the entries of $\overline{x}$ that are below the threshold $\frac{\alpha\opt}{d}$
are set to $0$. 

The guarantee of Algorithm \ref{alg:packing} is stated below.
\begin{thm}
\label{thm:packing}Given $\alpha,\beta\in(0,1)$, Algorithm \ref{alg:packing}
is $(\epsilon,\delta)$-differentially private and finds a solution
$x$ such that $1^{\top}x\ge\left(1-\alpha\right)\opt$ and with probability
at least $1-\beta$, $A_{i}x\le1+\alpha$ for all $i\in[n]$ except
for at most $s$ constraints, where

1. Without pre- and post-processing:
\begin{align*}
s & =O\left(\frac{\left(A_{\max}\opt\right)^{1.5}(\log d+\log\frac{n}{\alpha\beta})\sqrt{\log n\log\frac{1}{\delta}}}{\alpha^{2}\epsilon}\right).
\end{align*}

2. With pre- and post-processing:
\begin{align*}
s & =O\left(\frac{d^{1.5}(\log d+\log\frac{n}{\alpha\beta})\sqrt{\log n\log\frac{1}{\delta}}}{\alpha^{3.5}\epsilon}\right).
\end{align*}

\end{thm}

\subsection{Proof of Theorem \ref{thm:packing}}

\paragraph{Privacy guarantee.}

We have the following lemma.
\begin{lem}
Algorithm \ref{alg:packing} is $\left(\epsilon,\delta\right)$-DP.
\end{lem}

\begin{proof}
In each iteration, we use the Exponential mechanism with privacy parameter
$\epsilon'$ to find $\Delta_{t}$. By strong composition over $T$
iterations, Algorithm \ref{alg:packing} is $\left(\epsilon'\sqrt{2T\log\frac{1}{\delta}}+\frac{T\epsilon'^{2}}{2},\delta\right)$-DP.
For $\epsilon'=\frac{\epsilon}{2\sqrt{T\log\frac{1}{\delta}}}$ and
constant $\epsilon$, this implies Algorithm \ref{alg:packing} is
$\left(\epsilon,\delta\right)$-DP.
\end{proof}

\paragraph{Utility Analysis.}

First, we have the following guarantee of $\po$.
\begin{lem}
\label{lem:EM-gaurantee-packing}With probability at least $1-\beta$,
for all $t\in[T]$:
\begin{align*}
\left\langle \nabla\smax^{U}\left(\eta Ax_{t-1}\right),A\Delta_{t}\right\rangle  & \le1+\frac{\alpha}{10}.
\end{align*}
\end{lem}

\begin{proof}
We calculate the sensitivity of the score function $Q(j)=-\left\langle \nabla\smax^{U}\left(\eta Ax\right),A1_{j}\right\rangle \opt$.
As a reminder, we let $H$ be the upper bound on the entries of $A$.
If the algorithm does not do pre-processing, we simply have $H=S$.
Otherwise, $H=\frac{2d}{\alpha\cdot\opt}$.

Note that the coordinate of $\nabla\smax^{U}$ is bounded by $U=\frac{1}{s}$
and the entries of $A$ are bounded by $\frac{2d}{\alpha\cdot\opt}$.
For any $x$ and any neighboring inputs $A$ and $A'$ (after the
preproccessing), 
\begin{align*}
 & \left|\left\langle \nabla\smax^{U}\left(\eta Ax\right),A1_{j}\right\rangle -\left\langle \nabla\smax^{U}\left(\eta A'x\right),A'1_{j}\right\rangle \right|\\
\le\  & \left|\left\langle \nabla\smax^{U}\left(\eta Ax\right),A1_{j}\right\rangle -\left\langle \nabla\smax^{U}\left(\eta Ax\right),A'1_{j}\right\rangle \right|\\
 & \ +\left|\left\langle \nabla\smax^{U}\left(\eta Ax\right),A'1_{j}\right\rangle -\left\langle \nabla\smax^{U}\left(\eta A'x\right),A'1_{j}\right\rangle \right|\\
\le\  & \frac{3H}{s}.
\end{align*}
Hence the sensitivity of $Q$ is $\frac{3H\cdot\opt}{s}$. Since there
is a solution $x^{*}$ which satisfies $1^{\top}x^{*}=\opt$ and $Ax^{*}\le1$,
there must exist a coordinate $j$ such that $\left\langle \nabla\smax^{U}\left(\eta Ax_{t-1}\right),A1_{j}\right\rangle \opt\le1$.
There are $d$ coordinates, so with probability $1-\frac{\beta}{T}$,
in each iteration, the Exponential mechanism guarantees
\begin{align*}
\left\langle \nabla\smax^{U}\left(\eta Ax_{t-1}\right),A\Delta_{t}\right\rangle  & \le1+\frac{6H\cdot\opt\left(\log d+\log\frac{T}{\beta}\right)}{s\epsilon'}.
\end{align*}
Since $s=\frac{60H\cdot\opt\left(\log d+\log\frac{T}{\beta}\right)}{\alpha\epsilon'}$,
we have $\left\langle \nabla\smax^{U}\left(\eta Ax_{t-1}\right),A\Delta_{t}\right\rangle \le1+\frac{\alpha}{10}.$
The lemma statement can be obtained by taking the union bound over
$T$ iterations.
\end{proof}

\begin{lem}
With probability $1-\beta$, the output $\overline{x}$ of Algorithm
\ref{alg:packing} satisfies $1^{\top}\overline{x}\ge\left(1-\alpha\right)\opt$
and $A_{i}\overline{x}\le1+\alpha$ except for at most $s$ constraints
where 

1. Without pre- and post-processing:
\begin{align*}
s & =O\left(\frac{\left(A_{\max}\opt\right)^{1.5}(\log d+\log\frac{n}{\alpha\beta})\sqrt{\log n\log\frac{1}{\delta}}}{\alpha^{2}\epsilon}\right).
\end{align*}

2. With pre- and post-processing:
\begin{align*}
s & =O\left(\frac{d^{1.5}(\log d+\log\frac{n}{\alpha\beta})\sqrt{\log n\log\frac{1}{\delta}}}{\alpha^{3.5}\epsilon}\right).
\end{align*}

\end{lem}

\begin{proof}
First, we show that $A_{i}\overline{x}\le1+\alpha$ except for at
most $s$ constraints. We have$\left\Vert A\Delta_{t}\right\Vert _{\infty}\le\max_{i,j}A_{ij}\cdot\opt\le H\cdot\opt$.
Hence, for $\eta=\frac{\alpha}{10H\cdot\opt}$ and for all $t$, $\left\Vert \eta A\Delta_{t}\right\Vert _{\infty}\le\frac{\alpha}{10}\le1$.
Using property \eqref{eq:smax-update} of the $\smax^{U}$ function
\begin{align*}
\text{smax}^{U}\left(\eta Ax_{T}\right) & \leq\text{smax}^{U}\left(\eta Ax_{T-1}\right)+\left(1+\left\Vert \eta A\Delta_{t}\right\Vert _{\infty}\right)\left\langle \nabla\text{smax}^{U}\left(\eta Ax_{T-1}\right),\eta A\Delta_{T}\right\rangle \\
 & \le\text{smax}^{U}\left(\eta Ax_{T-1}\right)+\left(1+\frac{\alpha}{10}\right)\left\langle \nabla\text{smax}^{U}\left(\eta Ax_{T-1}\right),\eta A\Delta_{T}\right\rangle .
\end{align*}
Unrolling this over $T$ iterations, we obtain 
\begin{align*}
\text{smax}^{U}\left(\eta Ax_{T}\right) & \le\text{smax}^{U}\left(0\right)+\eta\left(1+\frac{\alpha}{10}\right)\sum^{T-1}_{t=1}\left\langle \nabla\text{smax}^{U}\left(\eta Ax_{t}\right),A\Delta_{t+1}\right\rangle \\
 & \overset{(i)}{\le}\ln n+\eta\left(1+\frac{\alpha}{10}\right)\cdot T\cdot\left(1+\frac{\alpha}{10}\right)\overset{(ii)}{\le}\eta T\left(1+\alpha\right).
\end{align*}
where $(i)$ comes from Lemma \ref{lem:EM-gaurantee-packing} and
$(ii)$ is due to the choice $\eta=\frac{\alpha}{10H\cdot\opt}$ and
$T=\frac{2\ln n}{\eta\alpha}.$ We then have
\begin{align*}
\max_{S\in\left[n\right]:\left|S\right|=s=\frac{1}{U}}\left\langle \frac{1_{S}}{\left|S\right|},\eta Ax_{T}\right\rangle  & \leq\max_{r\in\Delta_{n},r\leq U}\left\langle r,\eta Ax_{T}\right\rangle -\omega\left(r\right)\\
 & =\text{smax}^{U}\left(\eta Ax_{T}\right)\\
 & \le\eta T\left(1+\alpha\right).
\end{align*}
Therefore, we have that
\[
\max_{S\in\left[n\right]:\left|S\right|=s}\left\langle \frac{1_{S}}{\left|S\right|},A\overline{x}\right\rangle \le\max_{S\in\left[n\right]:\left|S\right|=s}\left\langle \frac{1_{S}}{\left|S\right|},A\frac{x_{T}}{T}\right\rangle \leq1+\alpha.
\]
This means among $s$ constraints with the largest values $A_{i}\overline{x}$,
the smallest value is upperbounded by $1+\alpha$. This implies all
but at most $s$ constraints $A_{i}x\le1+\alpha$ are satisfied, where
\begin{align*}
s & =O\left(\frac{H\cdot\opt\left(\log d+\log\frac{T}{\beta}\right)}{\alpha\epsilon'}\right).
\end{align*}

\textbf{Without pre- and post-processing}. We have
\begin{align*}
s & =O\left(\frac{S\cdot\opt\left(\log d+\log\frac{T}{\beta}\right)}{\alpha\epsilon'}\right)=O\left(\frac{\left(A_{\max}\opt\right)^{1.5}\left(\log d+\log\frac{n}{\alpha\beta}\right)\sqrt{\log n\log\frac{1}{\delta}}}{\alpha^{2}\epsilon}\right),
\end{align*}
where $S=O\left(A_{\max}\right)$. Further, since $\Delta_{t}=1_{j}\opt$
for some coordinate $j$, we have $1^{\top}\Delta_{t}=\opt$. Hence,
$\overline{x}$ satisfies: $1^{\top}\overline{x}=\frac{1}{T}\sum_{t}1^{\top}\Delta_{t}=\opt$,
as needed.

\textbf{With pre- and post-processing}. We have
\begin{align*}
s & =O\left(\frac{d\left(\log d+\log\frac{T}{\beta}\right)}{\alpha^{2}\epsilon'}\right)=O\left(\frac{d^{1.5}\left(\log d+\log\frac{n}{\alpha\beta}\right)\sqrt{\log n\log\frac{1}{\delta}}}{\alpha^{3.5}\epsilon}\right).
\end{align*}
For the objective, we also have, $\frac{1}{T}\sum_{t}1^{\top}\Delta_{t}=\opt.$
The output of the algorithm $\overline{x}$ is obtained by truncating
the entries of $\frac{1}{T}\sum_{t}1^{\top}\Delta_{t}$ that are smaller
than $\frac{\alpha\opt}{d}$. Therefore 
\begin{align*}
1^{\top}x & \ge\opt-d\cdot\frac{\alpha\opt}{d}=\left(1-\alpha\right)\opt.
\end{align*}
Next we show that at most $s$ constraints before truncation are not
satisfied. To distinguish between the input before and after the truncation,
let us denote by $A^{\text{init}}$ the original input and $A$ the
truncated input. Let $i\in[n]$ be such that $A_{i}\overline{x}\le1+\alpha$.
For $j\in[d]$ such that $A^{\text{init}}_{ij}\le\frac{2d}{\alpha\cdot\opt}$,
we have $A^{\text{init}}_{ij}=A_{ij}$, so $A^{\text{init}}_{ij}\overline{x}_{j}=A_{ij}\overline{x}_{j}$.
For $j\in[d]$ such that $A^{\text{init}}_{ij}>\frac{2d}{\alpha\cdot\opt}$,
we have $A_{ij}=\frac{2d}{\alpha\cdot\opt}$. It follows that $\overline{x}_{j}\le\frac{1+\alpha}{A_{ij}}\le\frac{\alpha\opt}{d}$,
which means $\overline{x}_{j}=0$. We then have $A^{\text{init}}_{ij}\overline{x}_{j}=0\le A_{ij}\overline{x}_{j}$.
Overall, $A^{\text{init}}_{i}\overline{x}\le A_{i}\overline{x}\le1+\alpha$.
Since at most $s$ truncated constraints are violated, it follows
that at most $s$ original constraints are violated.
\end{proof}

\section{Private Algorithm for Covering LP}

In this section, we seek to approximately solve the following problem
with differential privacy 
\[
\min\ 1^{\top}x\text{ s.t }Ax\ge1,x\ge0,
\]
where $A\in\R^{n\times d}_{\ge0}.$ We assume to know the value $\opt=\min\left\{ 1^{\top}x,Ax\ge1\right\} $,
so the goal is, given the approximation factor $\alpha$, to find
$x$ such that 
\[
1^{\top}x\le\left(1+\alpha\right)\opt\text{ and }Ax\ge1-\alpha.
\]
We also assume the entries of the input matrix $A$ are upper bounded
by a known value $R=O\left(A_{\max}\right)$, which can be found with
a binary search with only $\log$-dependence on the range of the entries
of $A$. 

\subsection{Algorithm}

\begin{algorithm}
\caption{Private Algorithm for Covering LP}

\label{alg:covering}

\textbf{Input}: $A\in\R^{n\times d}$, upper bound $R=O\left(A_{\max}\right)$
for the entries of $A$, optimal objective $\opt=\min\left\{ 1^{\top}x,Ax\ge1\right\} $,
$\alpha,\beta\in(0,1)$, privacy parameters $\epsilon,\delta$.

\begin{algorithmic}[1]

\STATE Let $H=\frac{40d}{\alpha\cdot\opt}$ if pre-processing else
$H=R$

\STATE \textbf{(Optional) Pre-processing}: for entry $A_{ij}$ in
$A$: $A_{ij}=\min\left\{ A_{ij},H\right\} $

\STATE Let $T=\frac{20H\cdot\opt\ln n}{\alpha},$ $\epsilon'=\frac{\epsilon}{2\sqrt{T\log\frac{1}{\delta}}}$,
$s=\frac{120H\cdot\opt\left(\log d+\log\frac{T}{\beta}\right)}{\alpha\epsilon'}$,
$U=\frac{1}{s}$

\STATE Initialize $x_{0}=0$

\STATE for $t=0,\dots,T-1$:

\STATE $\qquad$Let $\Delta_{t}=\co(x_{t},\epsilon')$

\STATE $\qquad$$x_{t+1}=x_{t}+\Delta_{t}$

\STATE \textbf{Output} $\overline{x}=\frac{x_{T}}{T}$

\STATE $\co(x,\epsilon)$:

\STATE $\qquad$Let $\eta=\frac{\alpha}{10H\cdot\opt}$

\STATE $\qquad$Use the Exponential mechanism with privacy parameter
$\epsilon$ to find coordinate $j\in[d]$ with score function $Q(i)=\left\langle \nabla\smin^{U}\left(\eta Ax\right),A1_{j}\right\rangle \opt$

\STATE $\qquad$\textbf{Output} $\Delta=1_{j}\opt$

\end{algorithmic}
\end{algorithm}

Our algorithm is presented in Figure \ref{alg:covering}. Similarly
to Algorithm \ref{alg:packing}, we also start with an optional preprocessing
step where large entries of the input matrix $A$ are truncated by
threshold $H=\frac{40d}{\alpha\cdot\opt}$. This step is required
only when we want to obtain a data-independent guarantee. Next, the
algorithm proceeds iteratively. In each iteration, the algorithm calls
$\co$, which uses the Exponential mechanism to find an update vector
$\Delta_{t}.$ This vector $\Delta_{t}$ satisfies $1^{\top}\Delta_{t}=\opt$
and with high probability $\left\langle \nabla\smin^{U}\left(\eta Ax_{t-1}\right),A\Delta_{t}\right\rangle \ge1-O\left(\alpha\right)$,
with the parameters defined in the algorithm. The output of the algorithm
is the average of the cumulative updates, i.e, $\overline{x}=\frac{\sum_{t}\Delta_{t}}{T}$. 

The guarantee of Algorithm \ref{alg:covering} is stated below. The
proof follows similarly to that of Theorem \ref{thm:packing} which
we defer to Appendix \ref{sec:Proof-of-Theorem-covering}.
\begin{thm}
\label{thm:covering}Given $\alpha,\beta\in(0,1)$, Algorithm \ref{alg:covering}
is $(\epsilon,\delta)$-differentially private and finds a solution
$x$ such that $1^{\top}x\le\left(1+\alpha\right)\opt$ and with probability
at least $1-\beta$, $A_{i}x\ge1-\alpha$ for all $i\in[n]$ except
for at most $s$ constraints, where

1. Without pre-processing:
\begin{align*}
s & =O\left(\frac{\left(A_{\max}\opt\right)^{1.5}\log d\log\frac{n}{\alpha\beta}\sqrt{\log n\log\frac{1}{\delta}}}{\alpha^{2}\epsilon}\right).
\end{align*}

2. With pre-processing:
\begin{align*}
s & =O\left(\frac{d^{1.5}\log d\log\frac{n}{\alpha\beta}\sqrt{\log n\log\frac{1}{\delta}}}{\alpha^{3.5}\epsilon}\right).
\end{align*}

\end{thm}

\section{Mixed Packing-Covering LP with DP}

In this section, we seek to approximately solve the following problem
with differential privacy: Find $x\ge0$ such that
\[
Px\le1\text{ and }Cx\ge1,
\]
where $P\in\R^{p\times d}_{\ge0}$, $C\in\R^{c\times d}_{\ge0}$.
The goal is, given the approximation factor $\alpha$, to find $x$
such that 
\[
Px\le1+\alpha\text{ and }Cx\ge1-\alpha.
\]
Similar to the previous cases, by Remark \ref{rem:assumption}, we
assume that the maximum entry in each column of $P$ and $C$ come
from a bounded range $[m,M]$ for known values $M>m>0$. By a binary
search, we can find upper bounds $S$ for the entries of $P$ and
$R$ for the entries of $C$ such that $S=O\left(P_{\max}\right)$
and $R=O\left(C_{\max}\right)$. We also assume to know a value $V=1^{\top}x$
for a feasible solution $x$.

We give two algorithms to solve this problem: one with guarantee depending
on the input data and one with guarantee only depending on the problem
dimensions. 
\begin{thm}
\label{thm:mixed}Given $\alpha,\beta\in(0,1)$ and let $n=p+c$.
Assuming that the maximum entry in each column of $P$ and $C$ lie
in range $[m,M]$ for known values $M>m>0$, we have the following

1. Algorithm \ref{alg:mixed} is $(\epsilon,\delta)$-differentially
private and finds a solution $x$ such that with probability at least
$1-\beta$, $P_{i}x\le1+\alpha$ for all $i\in[p]$ and $C_{i}x\ge1-\alpha$
for all $i\in[c]$, except for at most $s$ constraints, where 
\begin{align*}
s & =O\left(\frac{P_{\max}C_{\max}\sqrt{P_{\max}+C_{\max}}V^{2.5}}{\alpha^{4.5}\epsilon}\cdot\left(\log d+\log\frac{n}{\alpha\beta}\right)\sqrt{\log\frac{1}{\delta}\log n}\right),
\end{align*}
and $V=1^{\top}x$ for some feasible solution $x$ of the LP. 

2. Algorithm \ref{alg:mixed-1} is $(\epsilon,\delta)$-differentially
private and finds a solution $x$ such that with probability at least
$1-\beta$, $P_{i}x\le1+\alpha$ for all $i\in[p]$ and $C_{i}x\ge1-\alpha$
for all $i\in[c]$, except for at most $s$ constraints, where
\begin{align*}
s & =O\left(\frac{d^{3}\left(\log d+\log\frac{n}{\alpha\beta}\right)\sqrt{\log\frac{1}{\delta}\log n}}{\alpha^{6}\epsilon}+\frac{d^{2}\left(\log\log\frac{M}{m}+\log\frac{d}{\beta}\right)\log\frac{M}{m}}{\epsilon}\right).
\end{align*}

\end{thm}

The higher order of dependence on the problem parameters/dimension
and the approximation factor of both algorithms compared with the
counterparts for pure LPs is expected, due to the more complex interplay
between the packing and covering constraints. Even when privacy is
not considered, mixed packing-covering LPs require more iterations
for MWU algorithms to find an approximate solution.

\subsection{Algorithm with data-dependent guarantee}

\begin{algorithm}
\caption{Private Algorithm with Data-Dependent Guarantee for Mixed Packing-Covering LP}

\label{alg:mixed}

\begin{algorithmic}[1]

\STATE \textbf{Input}: $P\in\R^{p\times d}_{\ge0}$, $C\in\R^{c\times d}_{\ge0}$,
$\alpha,\beta\in(0,1)$, privacy parameters $\epsilon,\delta$, $V=1^{\top}x$
for some solution $x$

\STATE Initialize $x_{0}=0$

\STATE Let $C_{ij}\gets C_{ij}+\frac{\alpha}{V}$ for all $i,j$

\STATE Let $T=\Theta\left(\left(S+R\right)V\frac{\ln n}{\alpha^{3}}\right),$
$\epsilon'=\frac{\epsilon}{2\sqrt{T\log\frac{1}{\delta}}}$,

\STATE Let $s=O\left(\frac{SR\sqrt{S+R}V^{2.5}(\log d+\log\frac{T}{\beta})\sqrt{\log\frac{1}{\delta}\log n}}{\alpha^{4.5}\epsilon}\right)$,
$U=\frac{1}{s}$

\STATE for $t=0,\dots,T-1$:

\STATE $\qquad$Let $\Delta_{t}=\pc(x_{t},\epsilon')$

\STATE $\qquad$$x_{t+1}=x_{t}+\Delta_{t}$

\STATE Use the Exponential mechanism with privacy parameter $\epsilon'$
to find $k$ such that $k\cdot x_{T}$ minimizes the number of violated
constraints for $k$ being power of $1+\alpha$ in range $[\frac{m}{\alpha TM},\frac{60M}{\alpha m}]$.

\STATE \textbf{Output} $\overline{x}=k\cdot x_{T}$

\STATE $\pc(x,\epsilon)$:

\STATE $\qquad$Use the Exponential mechanism with privacy parameter
$\epsilon$ to find coordinate $j\in[d]$ with score function $Q(j)=-\frac{\left\langle \nabla\smax^{U}\left(Px\right),P1_{j}\right\rangle }{\left\langle \nabla\smin^{U}\left(Cx\right),C1_{j}\right\rangle }$

\STATE $\qquad$\textbf{Output} $\Delta=\frac{\alpha\cdot1_{j}}{30(S+R)}\cdot$

\end{algorithmic}
\end{algorithm}

Algorithm \ref{alg:mixed} uses $\pc$ to find the coordinate $j$
so that the relative increase between the packing and covering constraints,
characterized by $\frac{\left\langle \nabla\smax^{U}\left(Px\right),P1_{j}\right\rangle }{\left\langle \nabla\smin^{U}\left(Cx\right),C1_{j}\right\rangle }$
is minimized (i.e, we want the increase of the packing constraints
is small compared with that of the covering constraints). Similarly
to the cases of pure LPs, we use the exponential mechanism to find
this coordinate privately with score function $Q(j)=-\frac{\left\langle \nabla\smax^{U}\left(Px\right),P1_{j}\right\rangle }{\left\langle \nabla\smin^{U}\left(Cx\right),C1_{j}\right\rangle }$.
However, the sensitivity of this score function can be large since
the term $\left\langle \nabla\smin^{U}\left(Cx\right),C1_{j}\right\rangle $
can be as small as $C_{\min}$. To avoid the case when $C_{\min}$
is close to $0$, we perturb the entries of $C$ by adding to them
$\frac{\alpha}{V}$. While maintaining that $1^{\top}x=V$, the additional
term $\frac{\alpha}{V}\cdot1^{\top}x$ only contributes $\alpha$
to the covering constraints, hence if the modified covering constraints
are satisfied, the original are also satisfied approximately.

Note that $x_{T}$ guarantees that the ratio $\frac{\text{smax}^{U}\left(Px_{T}\right)}{\text{smin}^{U}\left(Cx_{T}\right)}\le1+\alpha.$
A final step before outputting the cumulated solution $x_{T}$ is
to find the scale $k$ so that $\overline{x}=k\cdot x_{T}$ so that
the constraints are satisfied. This can be done by using the exponential
mechanism in combination with a binary search.

We defer the analysis to Appendix \ref{sec:Proof-of-Theorem-mixed-1}.

\subsection{Algorithm with data-independent guarantee}

\begin{algorithm}
\caption{Private Algorithm with Data-Independent Guarantee Mixed Packing-Covering LP}

\label{alg:mixed-1}

\begin{algorithmic}[1]

\STATE \textbf{Input}: $P\in\R^{p\times d}_{\ge0}$, $C\in\R^{c\times d}_{\ge0}$,
$\alpha,\beta\in(0,1)$, privacy parameters $\epsilon,\delta$, known
range $[m,M]$ for $\max_{j}\left\{ P_{ij}\right\} $, $\max_{j}\left\{ C_{ij}\right\} $.

\STATE \textbf{Pre-processing}:

\STATE $\qquad$for $j\in[d]:$

\STATE $\qquad\qquad$Estimate $M_{j}=\max_{i}\left\{ P_{ij}\right\} $
and filter out $i$ such that $P_{ij}\ge M_{j}$ using $\est\left(m,M,\left\{ P_{ij}\right\} ^{p}_{j=1},\frac{\epsilon}{2d},\frac{\beta}{2d}\right)$

\STATE $\qquad$Let $C_{ij}\gets\min\left\{ C_{ij}+\frac{\alpha M_{j}}{d},\frac{40dM_{j}}{\alpha}\right\} $,
for all $i,j$

\STATE Initialize $x_{0}=0$,

\STATE Let $T=\Theta\left(\frac{d^{2}\ln n}{\alpha^{4}}\right),$
$\epsilon'=\frac{\epsilon}{4\sqrt{T\log\frac{1}{\delta}}}$,

\STATE Let $s=O\left(\frac{d^{3}\left(\log d+\log\frac{T}{\beta}\right)\sqrt{\log\frac{1}{\delta}\log n}}{\alpha^{6}\epsilon}\right)$,
$U=\frac{1}{s}$

\STATE for $t=0,\dots,T-1$:

\STATE $\qquad$Let $\Delta_{t}=\pc(x_{t},\epsilon')$

\STATE $\qquad$$x_{t+1}=x_{t}+\Delta_{t}$

\STATE Use the Exponential mechanism with privacy parameter $\epsilon'$
to find $k$ such that $k\cdot x_{T}$ minimizes the number of violated
constraints for $k$ being power of $1+\alpha$ in range $[\frac{m}{\alpha TM},\frac{60M}{\alpha m}]$.

\STATE \textbf{Output} $\overline{x}=k\cdot x_{T}$

\STATE $\est(m,M,\left\{ a_{i}\right\} ,\epsilon,\beta):$

\STATE $\qquad$Use the Exponential mechanism with privacy parameter
$\epsilon$ to find $k$ such that $2^{k}m\in[m,2M]$ with score function
$Q(k)=-\left|\left\{ a_{i}:a_{i}\ge2^{k}m\right\} \right|-\frac{2\left(\log\log\frac{2M}{m}+\log\frac{1}{\beta}\right)}{\epsilon}k$

\STATE $\qquad$\textbf{Output} $2^{k}m$.

\end{algorithmic}
\end{algorithm}

To obtain a bound independent of the input data, we need to pre-process
the data. As in the case of pure covering LPs, we need to clip the
large entries of the covering matrix. The challenge here is that the
clipping threshold has to take the packing constraints into account
and needs to be relative to the entries of the packing constraints.
For this reason, we need to estimate the maximum entries for each
coordinate of the packing matrix. This can be done relatively easy
using the exponential mechanism over the range of the entries $[m,M]$.
Once we have the estimate of $M_{j}=\max_{i}\left\{ P_{ij}\right\} $,
we can clip the entries of the covering matrix by $\frac{40dM_{j}}{\alpha}$,
besides perturbing them by $\frac{\alpha M_{j}}{d}$ as in Algorithm
\ref{alg:mixed}. The algorithm then proceeds iteratively, similarly
to Algorithm \ref{alg:mixed}. We give the full analysis in Appendix
\ref{sec:Analysis-of-Algorithm-mixed2}.

\bibliographystyle{plainnat}
\bibliography{ref}

\appendix

\section{Tools: Truncated Softmax and Softmin }

We expand the properties of $\smax^{U}$ and $\smin^{U}$ and their
proofs from Section \ref{sec:Preliminaries}. Recall that, given a
scalar parameter $0<U\leq1$, define $\mathcal{D}^{U}=\left\{ r\in\Delta_{n}:\max r\leq U\right\} $.
We define 
\[
\text{smax}^{U}\left(x\right)=\max_{r\in\mathcal{D}^{U}}\left\langle x,r\right\rangle -\omega\left(r\right),
\]
where $\omega\left(r\right)=\sum_{i}r_{i}\ln r_{i}$ is the negative
entropy. 
\begin{prop}
The gradient of $\smax^{U}\left(x\right)$ is given by
\begin{align*}
\left[\nabla\smax^{U}\left(x\right)\right]_{i} & =\min\left\{ U,\exp\left(x_{i}-t_{U}(x)\right)\right\} .
\end{align*}
\end{prop}

\begin{proof}
Using Lagrange multipliers we write
\begin{align*}
\max_{r\in\mathcal{D}^{U}}\left\langle x,r\right\rangle -\omega\left(r\right) & =\max_{r}\min_{y\geq0,z\geq0,\lambda}\left\langle x,r\right\rangle -\omega\left(r\right)+\left\langle y,r\right\rangle +\left\langle z,U-r\right\rangle +\lambda\left(1-\left\langle 1,r\right\rangle \right)\\
 & =\min_{y\geq0,z\geq0,\lambda}\max_{r}\left\langle x+y-z-\lambda1,r\right\rangle -\omega\left(r\right)+\left\langle z,U\right\rangle +\lambda.
\end{align*}
By first order optimality for the maximization problem we have
\[
x+y-z-\lambda1=1+\ln r
\]
and satisfy complementary slackness conditions $z_{i}\left(U-r\right)_{i}=0$,
$y_{i}r_{i}=0$.

We also note that by complementary slackness $y$ is nonzero only
in the degenerate case where all the entries of $x$ are either $U$
or $0$. Indeed, let $C=\left\{ i:r_{i}=U\right\} $, and let $\alpha=1-U\cdot\left|C\right|$
denote the remaining mass not on maximized entries. Provided that
$\alpha>0$, we can verify via a standard argument that $r$ will
put mass on all coordinates, otherwise it contradicts optimality.
Hence we only have to consider
\[
r=\exp\left(x-z-t\right)
\]
for some scalar $t$, and where $z_{i}\left(U-r\right)_{i}=0$. This
means that entries of $r$ smaller than $U$ satisfy $z_{i}=0$ and
hence this entries satisfy $r_{i}=\exp\left(x_{i}-t\right)$. Otherwise
$r_{i}=U$. This leads to the expression 
\[
\left[\nabla\text{smax}^{U}\left(x\right)\right]_{i}=r_{i}\in\left\{ U,\exp\left(x_{i}-t_{U}\left(x\right)\right)\right\} 
\]
with $t_{U}\left(x\right)$ chosen so that the resulting vector belongs
to the unit simplex.
\end{proof}

\begin{prop}
Let $C=\left\{ i:\left[\nabla\smax^{U}\left(x\right)\right]_{i}=U\right\} $
and $\overline{C}=\left[n\right]\setminus C$. We have 
\begin{align*}
\nabla^{2}\text{smax}^{U}\left(x\right) & \preceq\text{diag}\left(\left[\begin{array}{c}
\nabla\text{smax}^{U}\left(x\right)\bigg\vert_{\overline{C}}\\
0_{C}
\end{array}\right]\right).
\end{align*}
\end{prop}

\begin{proof}
It is easy to verify that $\nabla^{2}\text{smax}^{U}\left(x\right)$
must be $0$ on rows/columns corresponding to entries in $C$. Meanwhile
all the entries in $\nabla\text{smax}^{U}\bigg\vert_{\overline{C}}$
have $\alpha=1-U\cdot\left|C\right|$ total mass distributed among
them. Hence we can view $\nabla\text{smax}^{U}\bigg\vert_{\overline{C}}$
as a vector defined over $\alpha\Delta_{\left|\overline{C}\right|}$.
In this case we can verify that the gradient is just $g=\alpha\cdot\frac{\exp\left(x\right)}{\sum\exp\left(x\right)}$
and thus the Hessian satisfies the standard identity $H=\alpha\cdot\left(\text{diag}\left(g\right)-gg^{\top}\right)$.
Thus we have that 
\begin{align*}
\nabla^{2}\text{smax}^{U}\left(x\right) & =\text{diag}\left(\left[\begin{array}{c}
\nabla\text{smax}^{U}\left(x\right)\bigg\vert_{\overline{C}}\\
0_{C}
\end{array}\right]\right)-\frac{1}{\alpha}\cdot\left[\begin{array}{c}
\nabla\text{smax}^{U}\left(x\right)\bigg\vert_{\overline{C}}\\
0_{C}
\end{array}\right]\left[\begin{array}{c}
\nabla\text{smax}^{U}\left(x\right)\bigg\vert_{\overline{C}}\\
0_{C}
\end{array}\right]^{\top}\\
 & \preceq\text{diag}\left(\left[\begin{array}{c}
\nabla\text{smax}^{U}\left(x\right)\bigg\vert_{\overline{C}}\\
0_{C}
\end{array}\right]\right).\qedhere
\end{align*}
\end{proof}

\begin{prop}
For all $x$ and $u\in\R^{n}_{\ge0}$
\begin{align*}
\text{smax}^{U}\left(x+u\right) & \leq\text{smax}^{U}\left(x\right)+\frac{\exp\left(D\right)-1}{D}\left\langle \nabla\text{smax}^{U}\left(x\right),u\right\rangle ,
\end{align*}
where $D=\max_{i}\left\{ u_{i}\right\} $.
\end{prop}

\begin{proof}
Let $g\left(t\right)=\text{smax}^{U}\left(x+tu\right)$. Also let
$C\left(t\right)$ be the set of coordinates $i$ for which $\left[\nabla\text{smax}^{U}\left(x+tu\right)\right]_{i}=U$.
We have
\begin{align*}
g'\left(t\right) & =\left\langle \nabla\text{smax}^{U}\left(x+tu\right),u\right\rangle \\
g''\left(t\right) & =u^{\top}\nabla^{2}\text{smax}^{U}\left(x+tu\right)u\\
 & \leq u^{\top}\text{diag}\left(\left[\begin{array}{c}
\nabla\text{smax}^{U}\left(x+tu\right)\bigg\vert_{\overline{C(t)}}\\
0_{C(t)}
\end{array}\right]\right)u\\
 & \leq D\cdot\left\langle \nabla\text{smax}^{U}\left(x+tu\right),u\right\rangle \\
 & =D\cdot g'\left(t\right).
\end{align*}
Thus,
\[
g''\left(t\right)\leq D\cdot g'\left(t\right).
\]
Note that $g'$ is continuous (recall that since we regularize with
a strongly convex regularizer, $\text{smax}^{U}$ is smooth and thus
$\nabla\text{smax}^{U}$ is continuous) and the inequality holds on
each of the intervals where $C\left(t\right)$ is constant. Thus $\ln g'\left(t\right)$
is also continuous and
\[
\frac{d}{dt}\ln g'\left(t\right)=\frac{g"\left(t\right)}{g'\left(t\right)}\leq D.
\]
Therefore,
\begin{align*}
\ln g'\left(t\right) & \leq\ln g'\left(0\right)+t\cdot D\\
g'\left(t\right) & \leq\exp\left(t\cdot D\right)\cdot g'\left(0\right).
\end{align*}
So we have 
\begin{align*}
g\left(1\right) & =g\left(0\right)+\int^{1}_{0}g'\left(t\right)dt\\
 & \leq g\left(0\right)+\int^{1}_{0}\exp\left(t\cdot D\right)g'\left(0\right)dt\\
 & =g\left(0\right)+g'\left(0\right)\cdot\frac{e^{D}-1}{D}.
\end{align*}
and therefore,
\[
\text{smax}^{U}\left(x+u\right)\leq\text{smax}^{U}\left(x\right)+\frac{\exp\left(D\right)-1}{D}\left\langle \nabla\text{smax}^{U}\left(x\right),u\right\rangle .\qedhere
\]
\end{proof}

\section{Analysis of Algorithm \ref{alg:covering}}\label{sec:Proof-of-Theorem-covering}

\paragraph{Privacy guarantee.}
\begin{lem}
Algorithm \ref{alg:covering} is $\left(\epsilon,\delta\right)$-DP.
\end{lem}

\begin{proof}
The proof is similar to the packing case.
\end{proof}

\paragraph{Utility Analysis.}

To show the number of constraints dropped, first, we have the following
guarantee of $\mathsf{\co}$.
\begin{lem}
\label{lem:EM-gaurantee-covering}With probability at least $1-\beta$,
for all $t\in[T]$:
\begin{align*}
\left\langle \nabla\smin^{U}\left(\eta Ax_{t-1}\right),A\Delta_{t}\right\rangle  & \ge1-\frac{\alpha}{10}.
\end{align*}
\end{lem}

To prove this lemma, we use the following claim which shows the change
in the LP if the algorithm truncates the input entries. Again, let
us denote by $A^{\init}$ and $A$ the original and the truncated
input matrices for distinction.
\begin{claim}
\label{claim:covering}There exists $y\ge0$ such that $1^{\top}y=\opt$
and $Ay\ge1-\frac{\alpha}{20}$. 
\end{claim}

\begin{proof}
Let $x^{*}$ be such that $1^{\top}x^{*}=\opt$ and $A^{\init}x\ge1$.
Let $y=\frac{1}{1+\alpha/20}\left(x^{*}+\frac{\alpha\opt}{20d}\right)$.
We have $1^{\top}y=\opt$. Consider constraint $i$. If there is some
$j$ such that $A^{\init}_{ij}\ge\frac{40d}{\alpha\opt},$we have
$A_{i}y\ge\frac{40d}{\alpha\opt}\cdot\frac{\alpha\opt}{20d(1+\alpha)}\ge1$.
Otherwise we have $A_{i}=A^{\init}_{i}.$ Hence 
\begin{align*}
A_{i}y & =A^{\init}_{i}y\ge A^{\init}_{i}\frac{x^{*}}{1+\alpha/20}\ge\frac{1}{1+\alpha/20}\ge1-\frac{\alpha}{20}.
\end{align*}
Therefore there exists $y\ge0$ such that $1^{\top}y=\opt$ and $Ay\ge1-\frac{\alpha}{20}$.
\end{proof}

\begin{proof}
We calculate the sensitivity of the score function $Q(i)=\left\langle \nabla\smin^{U}\left(\eta Ax\right),A1_{i}\right\rangle \opt$.
Note that the coordinate of $\nabla\smax^{U}$ is upperbounded by
$U=\frac{1}{s}$ and the entries of $A$ are bounded by $H$. For
any $x$ and any neighboring inputs $A$ and $A'$ (after prepoccessing),
\begin{align*}
 & \left|\left\langle \nabla\smin^{U}\left(\eta Ax\right),A1_{i}\right\rangle -\left\langle \nabla\smin^{U}\left(\eta A'x\right),A'1_{i}\right\rangle \right|\\
\le\  & \left|\left\langle \nabla\smin^{U}\left(\eta Ax\right),A1_{i}\right\rangle -\left\langle \nabla\smin^{U}\left(\eta Ax\right),A'1_{i}\right\rangle \right|\\
 & \ +\left|\left\langle \nabla\smin^{U}\left(\eta Ax\right),A'1_{i}\right\rangle -\left\langle \nabla\smin^{U}\left(\eta A'x\right),A'1_{i}\right\rangle \right|\\
\le\  & \frac{3H}{s}.
\end{align*}
Hence the sensitivity of $Q$ is $\frac{3H}{s}$. By Claim \ref{claim:covering},
there must exist a coordinate $j$ such that $A1_{j}\opt\ge1-\frac{\alpha}{20}$.
There are $d$ coordinates, so with probability $1-\frac{\beta}{T}$,
in each iteration, the Exponential mechanism guarantees
\begin{align*}
\left\langle \nabla\smin^{U}\left(\eta Ax_{t-1}\right),A\Delta_{t}\right\rangle  & \ge1-\frac{\alpha}{20}-\frac{6H\cdot\opt\left(\log d+\log\frac{T}{\beta}\right)}{s\epsilon'}\\
 & \ge1-\frac{\alpha}{10},
\end{align*}
where $s=\frac{120H\cdot\opt\left(\log d+\log\frac{T}{\beta}\right)}{\alpha\epsilon'}$.
The lemma statement can be obtained by taking the union bound over
$T$ iterations.
\end{proof}

\begin{lem}
With probability $1-\beta$, the output $\overline{x}$ of Algorithm
\ref{alg:covering} satisfies $1^{\top}\overline{x}=\opt$ and $A_{i}\overline{x}\ge1-\alpha$
except for at most $s$ constraints where

1. Without pre-processing:
\begin{align*}
s & =O\left(\frac{\left(A_{\max}\opt\right)^{1.5}\left(\log d+\log\frac{n}{\alpha\beta}\right)\sqrt{\log n\log\frac{1}{\delta}}}{\alpha^{2}\epsilon}\right).
\end{align*}

2. With pre-processing:
\begin{align*}
s & =O\left(\frac{d^{1.5}\left(\log d+\log\frac{n}{\alpha\beta}\right)\sqrt{\log n\log\frac{1}{\delta}}}{\alpha^{3.5}\epsilon}\right).
\end{align*}

\end{lem}

\begin{proof}
Since $\Delta_{t}=1_{j}\opt$ for some coordinate $j$, we have $1^{\top}\Delta_{t}=\opt$.
Hence, 
\begin{align*}
1^{\top}\overline{x} & =\frac{1}{T}\sum_{t}1^{\top}\Delta_{t}=\opt.
\end{align*}
Next we show that at most $s$ constraints are not satisfied. To distinguish
between the input before and after the truncation, let us denote by
$A^{\text{init}}$ the original input and $A$ the truncated input.
Since $A^{\init}_{i}\ge A_{i}$ for all $i$, we only need to show
that $A_{i}\overline{x}\ge1-\alpha$ except for at most $s$ constraints.

We have$\left\Vert A\Delta_{t}\right\Vert _{\infty}\le\max_{i,j}A_{ij}\cdot\opt\le H\cdot\opt$.
Hence, for $\eta=\frac{\alpha}{10H\cdot\opt}$ and for all $t$, $\left\Vert \eta A\Delta_{t}\right\Vert _{\infty}\le\frac{\alpha}{10}\le1$.
Using property \eqref{eq:smax-update} of the $\smax^{U}$ function
\begin{align*}
\text{\ensuremath{\smin}}^{U}\left(\eta Ax_{T}\right) & \ge\smin^{U}\left(\eta Ax_{T-1}\right)+\left(1-\left\Vert \eta A\Delta_{t}\right\Vert _{\infty}\right)\left\langle \nabla\text{\ensuremath{\smin}}^{U}\left(\eta Ax_{T-1}\right),\eta A\Delta_{T}\right\rangle \\
 & \ge\text{\ensuremath{\smin}}^{U}\left(\eta Ax_{T-1}\right)+\left(1-\frac{\alpha}{10}\right)\left\langle \nabla\text{\ensuremath{\smin}}^{U}\left(\eta Ax_{T-1}\right),\eta A\Delta_{T}\right\rangle .
\end{align*}
Unrolling this over $T$ iterations, we obtain 
\begin{align*}
\text{\ensuremath{\smin}}^{U}\left(\eta Ax_{T}\right) & \ge\text{\ensuremath{\smin}}^{U}\left(0\right)+\eta\left(1-\frac{\alpha}{10}\right)\sum^{T-1}_{t=1}\left\langle \nabla\text{\ensuremath{\smin}}^{U}\left(\eta Ax_{t}\right),A\Delta_{t+1}\right\rangle \\
 & \overset{(i)}{\ge}-\ln n+\eta\left(1-\frac{\alpha}{10}\right)\cdot T\cdot\left(1-\frac{\alpha}{10}\right)\\
 & \overset{(ii)}{\ge}\eta T\left(1-\alpha\right).
\end{align*}
where $(i)$ comes from Lemma \ref{lem:EM-gaurantee-covering} and
$(ii)$ is due to the choice $\eta=\frac{\alpha}{10H\cdot\opt}$ and
$T=\frac{2\ln n}{\eta\alpha}.$

We then have
\begin{align*}
\min_{S\in\left[n\right]:\left|S\right|=s=\frac{1}{U}}\left\langle \frac{1_{S}}{\left|S\right|},\eta Ax_{T}\right\rangle  & \ge\min_{r\in\Delta_{n},r\leq U}\left\langle r,\eta Ax_{T}\right\rangle +\omega\left(r\right)\\
 & =\text{\ensuremath{\smin}}^{U}\left(\eta Ax_{T}\right)\\
 & \ge\eta T\left(1-\alpha\right).
\end{align*}
Therefore, we have that
\[
\min_{S\in\left[n\right]:\left|S\right|=s}\left\langle \frac{1_{S}}{\left|S\right|},A\overline{x}\right\rangle =\min_{S\in\left[n\right]:\left|S\right|=s}\left\langle \frac{1_{S}}{\left|S\right|},A\frac{x_{T}}{T}\right\rangle \ge1-\alpha.
\]
This means among the top $s$ constraints with the smallest values
$A_{i}\overline{x}$, the largest value of them is lower bounded by
$1-\alpha$. This implies all but at most $s$ constraints $A_{i}x\ge1-\alpha$
are satisfied.

\textbf{Without pre-processing}:
\begin{align*}
s & =O\left(\frac{R\cdot\opt\left(\log d+\log\frac{T}{\beta}\right)}{\alpha\epsilon'}\right)=O\left(\frac{\left(A_{\max}\opt\right)^{1.5}\left(\log d+\log\frac{n}{\alpha\beta}\right)\sqrt{\log n\log\frac{1}{\delta}}}{\alpha^{2}\epsilon}\right).
\end{align*}

\textbf{With pre-processing}:
\begin{align*}
s & =O\left(\frac{d\left(\log d+\log\frac{T}{\beta}\right)}{\alpha^{2}\epsilon'}\right)=O\left(\frac{d^{1.5}\left(\log d+\log\frac{n}{\alpha\beta}\right)\sqrt{\log n\log\frac{1}{\delta}}}{\alpha^{3.5}\epsilon}\right).\qedhere
\end{align*}
\end{proof}

\section{Analysis of Algorithm \ref{alg:mixed}}\label{sec:Proof-of-Theorem-mixed-1}

\paragraph{Privacy guarantee.}
\begin{lem}
Algorithm \ref{alg:mixed} is $\left(\epsilon,\delta\right)$-DP.
\end{lem}

\begin{proof}
The proof is similar to the packing case.
\end{proof}

\paragraph{Utility Analysis.}

To show the number of constraints dropped, first, we have the following
guarantee of $\pc$.
\begin{lem}
\label{lem:EM-gaurantee-mixed}With probability at least $1-\beta$,
for all $t\in[T]$:
\begin{align*}
\frac{\left\langle \nabla\smax^{U}\left(Ax_{t-1}\right),P\Delta_{t}\right\rangle }{\left\langle \nabla\smin^{U}\left(Ax_{t-1}\right),C\Delta_{t}\right\rangle } & \le1+\frac{\alpha}{5}.
\end{align*}
\end{lem}

\begin{proof}
We calculate the sensitivity of the score function $Q(j)=-\frac{\left\langle \nabla\smax^{U}\left(Px\right),P1_{j}\right\rangle }{\left\langle \nabla\smin^{U}\left(Cx\right),C1_{j}\right\rangle }$.
We consider two cases: when the neighboring inputs differ by a covering
constraint and when the differ by a packing constraint.

Case 1: Suppose that neighboring inputs are given by $(P,C)$ and
$(P,C')$. Note that after pre-processing, we have $C_{\min}$ and
$C'_{\min}$ are bounded from below by $\alpha/V$. Thus we have
\begin{align*}
 & \left|\frac{\left\langle \nabla\smax^{U}\left(Px\right),P1_{j}\right\rangle }{\left\langle \nabla\smin^{U}\left(Cx\right),C1_{j}\right\rangle }-\frac{\left\langle \nabla\smax^{U}\left(Px\right),P1_{j}\right\rangle }{\left\langle \nabla\smin^{U}\left(C'x\right),C'1_{j}\right\rangle }\right|\\
=\  & \left|\left\langle \nabla\smax^{U}\left(Px\right),P1_{j}\right\rangle \right|\cdot\frac{\left|\left\langle \nabla\smin^{U}\left(Cx\right),C1_{j}\right\rangle -\left\langle \nabla\smin^{U}\left(C'x\right),C'1_{j}\right\rangle \right|}{\left\langle \nabla\smin^{U}\left(Cx\right),C1_{j}\right\rangle \left\langle \nabla\smin^{U}\left(C'x\right),C'1_{j}\right\rangle }\\
\le\  & S\frac{3R/s}{C_{\min}C'_{\min}}\le\frac{3SRV^{2}}{s\alpha^{2}}.
\end{align*}

Case 2: Suppose that neighboring inputs are given by $(P,C)$ and
$(P',C)$. We also have
\begin{align*}
\left|\frac{\left\langle \nabla\smax^{U}\left(Px\right),P1_{j}\right\rangle }{\left\langle \nabla\smin^{U}\left(Cx\right),C1_{j}\right\rangle }-\frac{\left\langle \nabla\smax^{U}\left(P'x\right),P'1_{j}\right\rangle }{\left\langle \nabla\smin^{U}\left(Cx\right),C1_{j}\right\rangle }\right|\le & \frac{3S}{sC_{\min}}\le\frac{3SRV^{2}}{s\alpha^{2}}.
\end{align*}
There exists a coordinate $j$ such that 
\begin{align*}
\frac{\left\langle \nabla\smax^{U}\left(Px_{t-1}\right),P1_{j}\right\rangle }{\left\langle \nabla\smin^{U}\left(Cx_{t-1}\right),C1_{j}\right\rangle } & \le1.
\end{align*}
Therefore, with probability at least $1-\frac{\beta}{T}$,
\begin{align*}
\frac{\left\langle \nabla\smax^{U}\left(Px_{t-1}\right),P1_{j}\right\rangle }{\left\langle \nabla\smin^{U}\left(Cx_{t-1}\right),C1_{j}\right\rangle } & \le1+\frac{6SRV^{2}(\log d+\log\frac{T}{\beta})}{s\alpha^{2}\epsilon'}\le1+\frac{\alpha}{5},
\end{align*}
by the choice of $s$.
\end{proof}

\begin{lem}
\label{lem:mixed-utility}With probability $1-\beta$, the output
$\overline{x}$ of Algorithm \ref{alg:mixed} satisfies $P_{i}\overline{x}\le1+\alpha$
and $C_{i}\overline{x}\ge1-\alpha$ except for at most $s$ constraints
where
\begin{align*}
s & =O\left(\frac{P_{\max}C_{\max}\sqrt{P_{\max}+C_{\max}}V^{2.5}}{\alpha^{4.5}\epsilon}\left(\log d+\log\frac{n}{\alpha\beta}\right)\sqrt{\log\frac{1}{\delta}\log n}\right).
\end{align*}
 
\end{lem}

\begin{proof}
By the update $\Delta_{t}=\frac{\alpha\cdot1_{i}}{30(P_{\max}+C_{\max})}$,
we have $\left\Vert C\Delta_{t}\right\Vert _{\infty}\le\frac{\alpha}{30}\le1$,
hence,
\begin{align*}
\smin^{U}\left(Cx_{T}\right) & \ge\text{\ensuremath{\smin}}^{U}\left(Cx_{T-1}\right)+\left(1-\left\Vert C\Delta_{t}\right\Vert _{\infty}\right)\left\langle \nabla\text{\ensuremath{\smin}}^{U}\left(Cx_{T-1}\right),C\Delta_{T}\right\rangle \\
 & \ge\text{\ensuremath{\smin}}^{U}\left(Cx_{T-1}\right)+\left(1-\frac{\alpha}{30}\right)\left\langle \nabla\text{\ensuremath{\smin}}^{U}\left(Cx_{T-1}\right),C\Delta_{T}\right\rangle \\
 & \ge\text{\ensuremath{\smin}}^{U}\left(0\right)+\left(1-\frac{\alpha}{4}\right)\sum^{T-1}_{t=1}\left\langle \nabla\text{\ensuremath{\smin}}^{U}\left(Cx_{t}\right),C\Delta_{t+1}\right\rangle .
\end{align*}
Similarly, we have $\left\Vert P\Delta_{t}\right\Vert _{\infty}\le\frac{\alpha}{30}\le1$,
hence,
\begin{align*}
\text{smax}^{U}\left(Px_{T}\right) & \leq\text{smax}^{U}\left(Px_{T-1}\right)+\left(1+\left\Vert P\Delta_{t}\right\Vert _{\infty}\right)\left\langle \nabla\text{smax}^{U}\left(Px_{T-1}\right),P\Delta_{T}\right\rangle \\
 & \le\text{smax}^{U}\left(Px_{T-1}\right)+\left(1+\frac{\alpha}{30}\right)\left\langle \nabla\text{smax}^{U}\left(Px_{T-1}\right),P\Delta_{T}\right\rangle \\
 & \le\text{smax}^{U}\left(Px_{T-1}\right)+\left(1+\frac{\alpha}{30}\right)\left(1+\frac{\alpha}{5}\right)\left\langle \nabla\text{\ensuremath{\smin}}^{U}\left(Cx_{T-1}\right),C\Delta_{T}\right\rangle \\
 & \le\text{smax}^{U}\left(0\right)+\left(1+\frac{\alpha}{4}\right)\sum^{T-1}_{t=1}\left\langle \nabla\text{\ensuremath{\smin}}^{U}\left(Cx_{t}\right),C\Delta_{t+1}\right\rangle ,
\end{align*}
where in the third inequality, we use Lemma \ref{lem:EM-gaurantee-mixed}.
Combining both inequalities, we have
\begin{align*}
\frac{\text{smax}^{U}\left(Px_{T}\right)}{\text{smin}^{U}\left(Cx_{T}\right)} & \leq\frac{\ln n+\left(1+\frac{\alpha}{4}\right)\sum^{T-1}_{t=1}\left\langle \nabla\text{\ensuremath{\smin}}^{U}\left(Cx_{t}\right),C\Delta_{t+1}\right\rangle }{-\ln n+\left(1-\frac{\alpha}{4}\right)\sum^{T-1}_{t=1}\left\langle \nabla\text{\ensuremath{\smin}}^{U}\left(Cx_{t}\right),C\Delta_{t+1}\right\rangle }\\
 & =\frac{\frac{2}{\left(1-\frac{\alpha}{4}\right)}\ln n-\frac{\left(1+\frac{\alpha}{4}\right)}{\left(1-\frac{\alpha}{4}\right)}\ln n+\left(1+\frac{\alpha}{4}\right)\sum^{T-1}_{t=1}\left\langle \nabla\text{\ensuremath{\smin}}^{U}\left(Cx_{t}\right),C\Delta_{t+1}\right\rangle }{-\ln n+\left(1-\frac{\alpha}{4}\right)\sum^{T-1}_{t=1}\left\langle \nabla\text{\ensuremath{\smin}}^{U}\left(Cx_{t}\right),C\Delta_{t+1}\right\rangle }\\
 & =\frac{\frac{2}{\left(1-\frac{\alpha}{4}\right)}\ln n}{-\ln n+\left(1-\frac{\alpha}{4}\right)\sum^{T-1}_{t=1}\left\langle \nabla\text{\ensuremath{\smin}}^{U}\left(Cx_{t}\right),C\Delta_{t+1}\right\rangle }+\frac{\left(1+\frac{\alpha}{4}\right)}{\left(1-\frac{\alpha}{4}\right)}.
\end{align*}
Note that we set $\Delta_{t}=\frac{\alpha}{30(P_{\max}+C_{\max})}\cdot1_{i}$,
we have 
\begin{align*}
\sum^{T-1}_{t=1}\left\langle \nabla\text{\ensuremath{\smin}}^{U}\left(Cx_{t}\right),C\Delta_{t+1}\right\rangle  & \ge\frac{\alpha C_{\min}\cdot T}{30(P_{\max}+C_{\max})}.
\end{align*}
By the choice $T=\frac{P_{\max}+C_{\max}}{C_{\min}}\frac{480\ln n}{\alpha^{2}}$
we also have
\begin{align*}
\ln n & \le\frac{\alpha}{8}\left(-\ln n+\left(1-\frac{\alpha}{4}\right)\frac{\alpha C_{\min}\cdot T}{30(P_{\max}+C_{\max})}\right).
\end{align*}
This implies 
\begin{align*}
\frac{\text{smax}^{U}\left(Px_{T}\right)}{\text{smin}^{U}\left(Cx_{T}\right)} & \le1+\alpha.
\end{align*}
It follows that 
\begin{align*}
\frac{\max_{S\in\left[p\right]:\left|S\right|=s=\frac{1}{U}}\left\langle \frac{1_{S}}{\left|S\right|},Px_{T}\right\rangle }{\min_{S'\in\left[c\right]:\left|S'\right|=s=\frac{1}{U}}\left\langle \frac{1_{S'}}{\left|S'\right|},Cx_{T}\right\rangle } & \le\frac{\text{smax}^{U}\left(Px_{T}\right)}{\text{smin}^{U}\left(Cx_{T}\right)}\le1+\alpha.
\end{align*}
This means that for all $i\in\left[p\right]\setminus S$ and $k\in\left[c\right]\setminus S'$
$x_{T}$ satisfied $\frac{P_{i}x_{T}}{C_{k}x_{T}}\leq1+\alpha$. If
we scale $\overline{x}=kx_{T}$ such that $\max_{i}P_{i}\overline{x}=1$,
we $C_{i}\overline{x}\ge1-\alpha$, except for at most $s$ constraint.
Since,
\begin{align*}
\frac{\alpha m}{60M}\le\max_{i}P_{i}x_{T} & \le\frac{\alpha TM}{30m},
\end{align*}
we have 
\begin{align*}
k & \in\left[\frac{m}{\alpha TM},\frac{60M}{\alpha m}\right].
\end{align*}
Using the exponential mechanism for $\left(1+\alpha\right)$ powers
to minimize the number of violated constraints, we have the number
of violated constraints is upper bounded by 
\begin{align*}
s & =O\left(\frac{SRV^{2}\log d\log\frac{T}{\beta}}{\alpha^{3}\epsilon'}\right)=O\left(\frac{SRV^{2}(\log d+\log\frac{T}{\beta})\sqrt{T\log\frac{1}{\delta}}}{\alpha^{3}\epsilon}\right)\\
 & =O\left(\frac{P_{\max}C_{\max}\sqrt{P_{\max}+C_{\max}}V^{2.5}(\log d+\log\frac{n}{\alpha\beta})\sqrt{\log\frac{1}{\delta}\log n}}{\alpha^{4.5}\epsilon}\right).\qedhere
\end{align*}
\end{proof}

\section{Analysis of Algorithm \ref{alg:mixed-1}}\label{sec:Analysis-of-Algorithm-mixed2}

First, we examine the guarantee of $\est(m,M,\left\{ a_{i}\right\} ,\epsilon,\beta)$. 
\begin{lem}
Given $m,M,\epsilon,\beta$, with probability at least $1-\beta,$$\est(m,M,\left\{ a_{i}\right\} ,\epsilon,\beta)$
outputs $K=2^{k}m$ such that $K\le4\max_{i}\left\{ a_{i}\right\} $
and $\left|\left\{ a_{i}:a_{i}>K\right\} \right|\le O\left(\frac{\left(\log\log\frac{M}{m}+\log\frac{1}{\beta}\right)\log\frac{M}{m}}{\epsilon}\right)$.
\end{lem}

\begin{proof}
Let $k^{*}$ be the smallest number such that $2^{k}m\ge\max\left\{ a_{i}\right\} $.
We have $2^{k^{*}}m<2\max\left\{ a_{i}\right\} $. We have $Q$ has
sensitivity $1$ and $\max_{k:2^{k}m\in[m,2M]}Q(k)\ge-\frac{2\left(\log\log\frac{2M}{m}+\log\frac{1}{\beta}\right)}{\epsilon}k^{*}$.
The guarantee of the exponential mechanism gives us: with probability
$\ge1-\beta$, 
\begin{align*}
Q(k)\ge & -\frac{2\left(\log\log\frac{2M}{m}+\log\frac{1}{\beta}\right)}{\epsilon}\left(k^{*}+1\right).
\end{align*}
Therefore $k\le k^{*}+1$ and $\left|\left\{ a_{i}:a_{i}>2^{k}m\right\} \right|\le\frac{2\left(\log\log\frac{2M}{m}+\log\frac{1}{\beta}\right)}{\epsilon}\left(k^{*}+1\right)\le\frac{2\left(\log\log\frac{2M}{m}+\log\frac{1}{\beta}\right)\left(\log\frac{2M}{m}+1\right)}{\epsilon}$.
\end{proof}

\begin{lem}
\label{lem:PP}The Pre-processing step in Algorithm is $\frac{\epsilon}{2}$-DP
and with probability at least $1-\frac{\beta}{2}$, the number of
filtered out constraints is bounded by 
\begin{align*}
 & O\left(\frac{d^{2}\left(\log\log\frac{M}{m}+\log\frac{d}{\beta}\right)\log\frac{M}{m}}{\epsilon}\right).
\end{align*}
\end{lem}

\begin{proof}
The privacy guarantee is obtained by composition of $d$ calls to
$\est$, each of which is $\frac{\epsilon}{2d}$-DP. The number of
filtered constraints by each coordinate is at most $O\left(\frac{d\left(\log\log\frac{M}{m}+\log\frac{d}{\beta}\right)\log\frac{M}{m}}{\epsilon}\right)$
with probability $\ge1-\frac{\beta}{2d}$. Across $d$ coordinate,
this number is bounded by $O\left(\frac{d^{2}\left(\log\log\frac{M}{m}+\log\frac{d}{\beta}\right)\log\frac{M}{m}}{\epsilon}\right)$
with probability $\ge1-\frac{\beta}{2}$.
\end{proof}

\paragraph{Privacy guarantee.}
\begin{lem}
Algorithm \ref{alg:mixed-1} is $\left(\epsilon,\delta\right)$-DP.
\end{lem}

\begin{proof}
The pre-processing step is $\frac{\epsilon}{2}$-DP by Lemma \ref{lem:PP}.
In each iteration, we use the Exponential mechanism with privacy parameter
$\epsilon'$ to find $\Delta_{t}$. By strong composition over $T$
iterations, Algorithm \ref{alg:mixed-1} is $\left(\epsilon'\sqrt{2T\log\frac{1}{\delta}}+\frac{T\epsilon'^{2}}{2},\delta\right)$-DP.
For $\epsilon'=\frac{\epsilon}{4\sqrt{T\log\frac{1}{\delta}}}$ and
constant $\epsilon$, this implies Algorithm \ref{alg:mixed-1} is
$\left(\epsilon,\delta\right)$-DP.
\end{proof}

\paragraph{Utility Analysis.}
\begin{lem}
\label{lem:EM-gaurantee-mixed-1}With probability at least $1-\beta$,
for all $t\in[T]$:
\begin{align*}
\frac{\left\langle \nabla\smax^{U}\left(Ax_{t-1}\right),P\Delta_{t}\right\rangle }{\left\langle \nabla\smin^{U}\left(Ax_{t-1}\right),C\Delta_{t}\right\rangle } & \le1+\frac{\alpha}{5}.
\end{align*}
\end{lem}

\begin{claim}
There exists $y\ge0$ such that $\sum_{i}M_{i}x_{i}\le d$, $Py\le1+\frac{\alpha}{20}$
and $Cy\ge\frac{1}{1+\alpha/20}$. 
\end{claim}

\begin{proof}
Let $x^{*}$ be such that $P^{\init}x^{*}\le1$ and $C^{\init}x\ge1$.

Let $y$ be such that $y_{j}=\frac{1}{1+\alpha/20}\left(x^{*}_{j}+\frac{\alpha}{20dM_{j}}\right)$.
We have $\sum_{j}M_{j}y_{j}\le\frac{1}{1+\alpha/20}\left(d+\frac{\alpha}{20}\right)\le d$.

Consider packing constraint $i$. 
\begin{align*}
P_{i}y & \le\frac{1}{1+\alpha/20}\left(1+\frac{\alpha}{20}\right)=1.
\end{align*}

Consider covering constraint $i$.

If there is some $j$ such that $C^{\init}_{ij}\ge\frac{40dM_{j}}{\alpha},$we
have $C_{i}y\ge\frac{1}{1+\alpha/20}\frac{40dM_{j}}{\alpha}\cdot\frac{\alpha}{20dM_{j}}\ge1$.
Otherwise we have $C_{i}\ge C^{\init}_{i}.$ Hence 
\begin{align*}
C_{i}y & =C^{\init}_{i}y\ge C^{\init}_{i}\frac{x^{*}}{1+\alpha/20}\ge\frac{1}{1+\alpha/20}.\qedhere
\end{align*}
\end{proof}

\begin{proof}
We calculate the sensitivity of the score function $Q(j)=-\frac{\left\langle \nabla\smax^{U}\left(Px\right),P1_{j}\right\rangle }{\left\langle \nabla\smin^{U}\left(Cx\right),C1_{j}\right\rangle }$.
We consider two cases: when the neighboring inputs differ by a covering
constraint and when the differ by a packing constraint.

Case 1: Suppose that neighboring inputs are given by $(P,C)$ and
$(P,C')$. Note that after pre-processing, we have $C_{\min}$ and
$C'_{\min}$ are bounded from below by $\frac{\alpha M_{j}}{d}$.
Thus we have
\begin{align*}
 & \left|\frac{\left\langle \nabla\smax^{U}\left(Px\right),P1_{j}\right\rangle }{\left\langle \nabla\smin^{U}\left(Cx\right),C1_{j}\right\rangle }-\frac{\left\langle \nabla\smax^{U}\left(Px\right),P1_{j}\right\rangle }{\left\langle \nabla\smin^{U}\left(C'x\right),C'1_{j}\right\rangle }\right|\\
=\  & \left|\left\langle \nabla\smax^{U}\left(Px\right),P1_{j}\right\rangle \right|\frac{\left|\left\langle \nabla\smin^{U}\left(Cx\right),C1_{j}\right\rangle -\left\langle \nabla\smin^{U}\left(C'x\right),C'1_{j}\right\rangle \right|}{\left\langle \nabla\smin^{U}\left(Cx\right),C1_{j}\right\rangle \left\langle \nabla\smin^{U}\left(C'x\right),C'1_{j}\right\rangle }\\
\le\  & \frac{M_{j}}{s}\frac{3\frac{40dM_{j}}{\alpha}}{\left(\frac{\alpha M_{j}}{d}\right)^{2}}=\frac{120d^{2}}{s\alpha^{3}}.
\end{align*}

Case 2: Suppose that neighboring inputs are given by $(P,C)$ and
$(P',C)$. We also have
\begin{align*}
\left|\frac{\left\langle \nabla\smax^{U}\left(Px\right),P1_{j}\right\rangle }{\left\langle \nabla\smin^{U}\left(Cx\right),C1_{j}\right\rangle }-\frac{\left\langle \nabla\smax^{U}\left(P'x\right),P'1_{j}\right\rangle }{\left\langle \nabla\smin^{U}\left(Cx\right),C1_{j}\right\rangle }\right|\le & \frac{3M_{j}}{s\left(\frac{\alpha M_{j}}{d}\right)}\le\frac{3d}{s\alpha}\le\frac{120d^{2}}{s\alpha^{3}}.
\end{align*}
There must exist a coordinate $j$ such that 
\begin{align*}
\frac{\left\langle \nabla\smax^{U}\left(Px_{t-1}\right),P1_{j}\right\rangle }{\left\langle \nabla\smin^{U}\left(Cx_{t-1}\right),C1_{j}\right\rangle } & \le\left(1+\frac{\alpha}{20}\right)^{2}.
\end{align*}
Therefore, with probability at least $1-\frac{\beta}{T}$,
\begin{align*}
\frac{\left\langle \nabla\smax^{U}\left(Px_{t-1}\right),P1_{j}\right\rangle }{\left\langle \nabla\smin^{U}\left(Cx_{t-1}\right),C1_{j}\right\rangle } & \le\left(1+\frac{\alpha}{20}\right)^{2}+\frac{120d^{2}\left(\log d+\log\frac{T}{\beta}\right)}{s\alpha^{3}\epsilon'}\le1+\frac{\alpha}{5}.\qedhere
\end{align*}
\end{proof}

\begin{lem}
With probability $1-\beta$, the output $\overline{x}$ of Algorithm
\ref{alg:mixed} satisfies $P_{i}\overline{x}\le1+\alpha$ and $C_{i}\overline{x}\ge1-\alpha$
except for at most $r$ constraints where
\begin{align*}
r & =O\left(\frac{d^{3}\left(\log d+\log\frac{n}{\alpha\beta}\right)\sqrt{\log\frac{1}{\delta}\log n}}{\alpha^{6}\epsilon}+\frac{d^{2}\left(\log\log\frac{M}{m}+\log\frac{d}{\beta}\right)\log\frac{M}{m}}{\epsilon}\right).
\end{align*}
\end{lem}

\begin{proof}
The proof is similar to the that of Lemma \ref{lem:mixed-utility}.
The number of violated constraints is 
\begin{align*}
s & =O\left(\frac{d^{2}\left(\log d+\log\frac{T}{\beta}\right)}{\alpha^{4}\epsilon'}\right)=O\left(\frac{d^{2}\left(\log d+\log\frac{T}{\beta}\right)\sqrt{T\log\frac{1}{\delta}}}{\alpha^{4}\epsilon}\right)\\
 & =O\left(\frac{d^{3}\left(\log d+\log\frac{n}{\alpha\beta}\right)\sqrt{\log\frac{1}{\delta}\log n}}{\alpha^{6}\epsilon}\right).
\end{align*}
Taken into account the number of filtered constraints, we obtain the
bound in the lemma.
\end{proof}

\end{document}